\documentclass[11pt]{amsart}

\usepackage[utf8]{inputenc} 

\newtheorem{theorem}{Theorem}[section]
\newtheorem{corollary}[theorem]{Corollary}

\newtheorem{lemma}[theorem]{Lemma}

\theoremstyle{definition}
\newtheorem{definition}[theorem]{Definition}
\newtheorem{remark}[theorem]{Remark}
\newtheorem{example}{Example}

\usepackage[dvipdfmx]{graphicx}

\usepackage{physics} 
\usepackage[version=4]{mhchem}

\title {Effective Kinetics of Chemical Reaction Networks}
\author{Tomoharu Suda}
\address{RIKEN Center for Sustainable Resource Science, RIKEN, Wako, Japan}
\email{tomoharu.suda@riken.jp}
\begin{document}
\maketitle
\begin{abstract}
Chemical reaction network theory is a powerful framework to describe and analyze chemical systems. While much about the concentration profile in an equilibrium state can be determined in terms of the graph structure, the overall reaction's time evolution depends on the network's kinetic rate function. In this article, we consider the problem of the effective kinetics of a chemical reaction network regarded as a conversion system from the feeding species to products. We define the notion of effective kinetics as a partial solution of a system of non-autonomous ordinary differential equations determined from a chemical reaction network. Examples of actual calculations include the Michaelis-Menten mechanism, for which it is confirmed that our notion of effective kinetics yields the classical formula. Further, we introduce the notion of straight-line solutions of non-autonomous ordinary differential equations to formalize the situation where a well-defined reaction rate exists and consider its relation with the quasi-stationary state approximation used in microkinetics. Our considerations here give a unified framework to formulate the reaction rate of chemical reaction networks.
\end{abstract}
\section{Introduction}
Chemical reaction network theory is a useful framework to describe chemical systems \cite{yu2018mathematical, feinberg2019foundations}. Once constituting reactions are specified, we may use it to deduce the dynamical properties of the system. Remarkably, knowing the topological properties of the network is often sufficient for determining the long-time behavior, and detailed knowledge of the reaction speed is not required \cite{feinberg1987chemical}. Here, we note that long-time behavior studied in the chemical reaction network theory is largely independent of time scale. This is because it is defined in terms of the limiting behavior of a dynamical system, which is actually an infinite-time behavior and, therefore, invariant under the rescaling of time. 

The question of kinetics is different. Here, the main concern is finite-time behavior, which is heavily dependent on the time scale and, therefore, beyond the scope of usual dynamical systems theory. Further, obtaining a mathematically well-defined notion of kinetics itself is not very straightforward. Here we recall the standard definition of the reaction rate, which is given via the time differential of the extent of the reaction, which is in turn usually defined only for a single reaction \cite{mcnaught1997compendium, Laidler1996}.
Chemical reaction networks represent composite reactions, and therefore, it is not obvious how the notion of reaction rate can be applied. In addition, we notice that there is much ambiguity as to what reaction is elementary. 

That said, the reaction rate for composite reactions has been considered for various chemical reaction networks, including the classical example of Michaelis-Menten kinetics \cite{houston2012chemical, cornish2013fundamentals}. Also in \cite{toth2023}, the notion of the reaction extent is generalized for composite reactions. The derivation for effective rate laws for the change of concentration usually involves the quasi-stationary state approximation (QSSA), where some of the reaction processes are assumed to be in a stationary state \cite{stiefenhofer1998quasi, turanyi1993error}. We notice that the principle working here is to regard a subnetwork consisting of intermediate steps as a black box conversion system and find the rate of conversion. The purpose of the present article is to propose a mathematical framework that generalizes this procedure of abstraction. The concept of kinetics can be defined as the conversion rate of one set of species to another by a chemical system. This is experimentally observable and,  in the case of elementary reactions, coincides with the mathematical notion of reaction rate. This is what we call effective kinetics here.

Naturally, our approach here is similar in spirit to the classical notion of QSSA and the reduction schemes of chemical reaction networks, in particular, the recent theory of chemical reaction circuits \cite{hangos2021structural, hirono2021structural, avanzini2022circuit}. The main difference is that we do not suppose the existence of a steady state. In this view, the kinetics of an effective reaction may depend not only on the concentration profile of input species but also on how they are supplied. This is most prominent in the case of the delay effect, which has been considered in connection with the signal processing theory and biochemical oscillators \cite{samoilov2002signal, moles2015delay, novak2008design}. When a system has a long chain of reactions, it takes some time for an input to have an observable effect on the output. Even when the reactant species are supplied so that the concentration is kept constant, the overall reaction may not have a well-defined rate. For example, if the system exhibits an oscillatory behavior, we cannot expect the existence of such a rate. That said, as we will see later, our notion of effective kinetics coincides with QSSA under some circumstances.

The rest of the article is organized as follows. First, in Section \ref{prelim}, we provide definitions of basic concepts to fix notation. Section \ref{eff} contains the definition of effective kinetics and calculation examples for simple systems, including the Michaelis-Menten mechanism.  Section \ref{rate} introduces the notion of straight-line solutions of an ordinary differential equation, which is useful in describing the problem of the existence of reaction rate. Finally, in Section \ref{skew}, we consider how our framework relates to the classical QSSA. Section \ref{conc} includes concluding remarks.

\section{Preliminaries}\label{prelim}
In this section, we introduce some basic terminology to fix notation.
\begin{definition}
Let $A=\{a_1, a_2, \cdots , a_N\}$ be a finite set. The set of all formal sums of the elements of $A$ with natural number coefficients is denoted by
\[
	\mathbb{N}^A := \{n _1 a_1+ n_2 a_2 + \cdots + n_N a_N \mid n_1, n_2, \cdots n_N \in \mathbb{N}\}.
\]
The set of all formal sums of the elements of $A$ with real number coefficients is denoted by $\mathbb{R}^A.$
\end{definition}

\begin{definition}[Reaction Networks]
A \textbf{chemical reaction network} is a quadruple $(\mathcal{S}, \mathcal{R}, \vb{s},\vb{t})$ such that
\begin{enumerate}
	\item $\mathcal{S}$ and $\mathcal{R}$ are finite sets.
	\item $\vb{s}$ and $\vb{t}$ are maps $\vb{s},\vb{t}: \mathcal{R} \to \mathbb{N}^\mathcal{S}$.
\end{enumerate}
Here $\mathcal{S}$ is the set of species, and $\mathcal{R}$ is the set of reactions. Each reaction $\rho \in \mathcal{R}$ has the form \ce{\vb{s}(\rho) -> \vb{t}(\rho)}. For each $\rho \in \mathcal{R}$, $\vb{s}(\rho)$ and $\vb{t}(\rho)$ are called \textbf{complexes}. By abuse of notation, we denote by $\vb{s}(\rho)$ the set of species $\alpha$ with a non-zero coefficient in $\vb{s}(\rho)$. The case for $\vb{t}(\rho)$ is similar.
\end{definition}
Kinetics of the chemical reaction networks are described using ordinary differential equations (ODEs). An important distinction is that of autonomous and non-autonomous equations. While most chemical reaction networks are considered isolated systems, which are modeled by autonomous equations, this is not the case for our situation of open chemical networks where there can be interactions between the outside and inside of systems.
\begin{definition}[Autonomous and non-autonomous ordinary differential equations]
Let $f_1, \cdots, f_n$ be functions on $\mathbb{R}^n$.
An equation of the form
\[
	\dv{x_i}{t} = f_i(x_1, x_2, \cdots, x_n) \qq\, i=1,2, \cdots, n
\]
is called an \textbf{autonomous ordinary differential equation}. If  $f_1, \cdots, f_n$ are functions on $\mathbb{R}\times \mathbb{R}^n$,
an equation of the form
\[
	\dv{x_i}{t} = f_i(t, x_1, x_2, \cdots, x_n)\qq\, i=1,2, \cdots, n
\]
is called a \textbf{non-autonomous ordinary differential equation}. The problem of finding functions $x_i(t)$ ($i=1,2,\cdots, n$) which satisfy the given initial condition $x_i(t_0) = a_i$ is called the \textbf{initial value problem} (IVP).
\end{definition}

In what follows, $C(X, Y)$ denotes the set of all continuous maps from $X$ to $Y$ where $X$ and $Y$ are assumed to be metric spaces.

\begin{definition}[Kinetic reaction networks]
A \textbf{kinetic chemical reaction network} is a quintuple $(\mathcal{S}, \mathcal{R}, \vb{s}, \vb{t}, \mathcal{K})$ such that
\begin{enumerate}
	\item $(\mathcal{S}, \mathcal{R}, \vb{s}, \vb{t})$ is a chemical reaction network.
	\item $\mathcal{K}$ is a map $\mathcal{R} \to C(\mathbb{R}^\mathcal{S}, \mathbb{R}_{\geq 0})$.
\end{enumerate}
\end{definition}
The dynamics of a kinetic reaction network is given by the following ODE:

\begin{equation}\label{eqn_crn}
	\dv{\vb{x}}{t} = \sum_{\rho \in \mathcal{R}} \mathcal{K}(\rho)(\vb{x}) \qty(\vb{t}(\rho) - \vb{s}(\rho)),
\end{equation}
where $\vb{x} \in \mathbb{R}^\mathcal{S}$ denotes the concentration vector of species. Here we note that equation (\ref{eqn_crn}) defines a dynamical system on $\mathbb{R}_{\geq 0}^\mathcal{S}$.

Given an autonomous ODE, we may modify the equation so that its response to external inputs is considered. 
\begin{definition}
Let 
\[
 \dv{x_i}{t} = f_i(x_1, x_2, \cdots, x_n) \qq\, i=1,2, \cdots, n
\]
 be an autonomous ODE. For an input function $\vb{I}: \mathbb{R} \to \mathbb{R}^k,$ the induced system is the non-autonomous ODE given by
 \[
 	\dv{x_i}{t} = f_i\qty(x_1, x_2, \cdots, x_{n-k}, I_1(t), I_2(t), \cdots, I_k(t)) \qq\, i=1,2, \cdots, n-k.
 \]

We define the \textbf{response} $R[\vb{I}; \vb{a}]:\mathbb{R} \to \mathbb{R}^n$ of the system under the input $\vb{I}$ by
 \[
 	R[\vb{I}; \mathbf{a}](t) := \qty(x_1(t), x_2(t), \cdots, x_n(t)),
 \]
 where $R[\vb{I}; \mathbf{a}](t)$ is the solution for IVP of these equations under the initial condition $\vb{x}(0) = \vb{a}$.
We call an index $i$ is \textbf{terminal} if $I_i(t) \neq 0$ for some $t$ and \textbf{internal} if $I_i(t) = 0$ for all $t.$
\end{definition}

\section{Effective kinetics of chemical reaction networks}\label{eff}
In this section, we introduce the concept of effective kinetics and present examples of calculations for some simple systems.
\subsection{Definition of effective kinetics}
Here, we consider a situation where some species are supplied from outside of the system or concentration is controlled. The dynamics of such systems are described as open chemical reaction networks \cite{baez2017compositional, horn1974open}. 
\begin{definition}
Let $(\mathcal{S}, \mathcal{R}, \vb{s}, \vb{t}, \mathcal{K})$ be a kinetic reaction network. The kinetics of an \textbf{open reaction network} with the terminal species $\mathcal{S}_b \subset \mathcal{S}$ and input $\vb{I}_b \in \mathbb{R}^{\mathcal{S}_b}$ is defined by
\[
		\dv{\vb{x}_i}{t} =  \sum_{\rho \in \mathcal{R}} \mathcal{K}(\rho)\qty(\vb{x}_i, \vb{I}_b(t)) \qty(t_i(\rho) - s_i(\rho)),
\]
where $\vb{x}_i$, $t_i$ and $s_i$ are defined by projections to $\mathbb{R}^{\mathcal{S}_i},$ where $\mathcal{S}_i = \mathcal{S} \backslash \mathcal{S}_b.$
\end{definition}
\begin{remark}
It is more common to define open chemical networks using input by addition. However, we do not take this view here because the main interest is to find the conversion rate. 
\end{remark}
We define the effective kinetics by a `partially integrated' solution.
\begin{definition}[Effective kinetics]
Let $(\mathcal{S}, \mathcal{R}, \vb{s}, \vb{t}, \mathcal{K})$ be a kinetic reaction network. The \textbf{effective kinetics} $\mathcal{E}_s$ of species $s \in \mathcal{S}$ under the input $\vb{I}_b$ of the terminal species $\mathcal{S}_b \subset \mathcal{S}$ and the initial concentration $\vb{a}$ is defined by
\begin{equation}\label{def_efkine}
	 \mathcal{E}_s[\vb{I}_b;\vb{a}](t) := \sum_{\rho \in \mathcal{R}} \mathcal{K}(\rho)(R[\vb{I}_b; \vb{a}]) P_s\qty(\vb{t}(\rho) - \vb{s}(\rho))
\end{equation}
 where $P_s: \mathbb{R}^\mathcal{S} \to \mathbb{R}$ is defined by the projection to the subspace $\mathbb{R}^{\{s\}}$.
\end{definition}
By definition, we  have
\[
	\dv{\vb{x}_s}{t} =  \mathcal{E}_s[\vb{I}_b;\vb{a}](t).
\]
\begin{remark}
In \cite{toth2023}, the reaction extent of a chemical reaction network is considered. While we do not consider the rate in terms of the reaction frequency, here we compare the definitions to clarify the relation between them. The reaction extent $\xi$ is defined for an autonomous system of the form (\ref{eqn_crn}) by
\[
	\xi_\rho(t) := V \int_{0}^{t} \mathcal{K}(\rho)\qty(\vb{x}(\tau)) \mathrm{d}\tau,
\]
where $V$ is the volume of the system. The differences from the effective kinetics defined by (\ref{def_efkine}) are the setting of the system (automonous vs. non-autonomous) and the choice of variables (occurrence of reactions vs. concentration).
\end{remark}
\begin{definition}[Average effective kinetics]
Let $(\mathcal{S}, \mathcal{R}, \vb{s}, \vb{t}, \mathcal{K})$ be a kinetic reaction network. The \textbf{average effective kinetics} of species $s \in \mathcal{S}$ under the input $\vb{I}_b$ of the terminal species $\mathcal{S}_b \subset \mathcal{S}$ and the initial concentration $\vb{a}$ is defined to be
\[
	\lim_{T \to \infty}\frac{1}{T}\int_0^T \mathcal{E}_s[\vb{I}_b;\vb{a}](t)\mathrm{d} t.
\]
\end{definition}
When the average effective kinetics exists, it makes sense to discuss the rate of overall reaction without mentioning time. We note that the average effective kinetics exists if there is a well-defined rate at which the effective kinetics converge. Later, we will consider this situation in connection with the quasi-stationary state approximation scheme.

Here, we note that the initial behavior does not affect the average effective kinetics.
\begin{lemma}
Let $(\mathcal{S}, \mathcal{R}, \vb{s}, \vb{t}, \mathcal{K})$ be a kinetic reaction network.  For all $T_0>0,$ we have
\[
	\lim_{T \to \infty}\frac{1}{T}\int_0^T \mathcal{E}_s[\vb{I}_b;\vb{a}](t)\mathrm{d} t=\lim_{T \to \infty}\frac{1}{T}\int_{T_0}^{T_0+T} \mathcal{E}_s[\vb{I}_b;\vb{a}](t)\mathrm{d} t.
\]
\end{lemma}

\subsection{Examples of calculation for simple systems}
\begin{example}[Michaelis-Menten]
As a first example, we consider the Michaelis-Menten mechanism of enzymatic reaction:

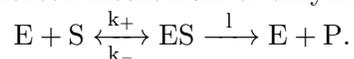
\ce{ E + S <->[k_+][k_-] ES ->[l] E +P}.

The governing equation is given by
\[
\begin{aligned}
	\dv{[S]}{t} & = -k_+[E][S] + k_- [ES]\\
	\dv{[E]}{t} & = -k_+[E][S] +(k_- + l)[ES]  \\
	\dv{[ES]}{t} &= k_+[E][S] -(k_- + l) [ES]\\
	\dv{[P]}{t} & = l [ES].
\end{aligned}
\]

When we consider the case with input, we ignore the equation for $[S]$ and set $[S](t) = s(t).$ Then, we obtain an effective kinetics of the form
\[
	\mathcal{E}_P[\vb{I}_b;\vb{a}](t) = l k_+  E_0 e^{-k_+ \int_0^t s(\tau) \mathrm{d} \tau -(k_- + l)t} \int_0^t e^{k_+ \int_0^r s(\tau) \mathrm{d} \tau+ (k_-  + l)r} s(r) \mathrm{d} r.
\]
 Here, the nonlinearity of the MM kinetics is reflected in the non-additivity of the response.

If we keep the concentration of \ce{S} constant in the MM kinetics, this expression is simplified and we have
\[
	\mathcal{E}_P[\vb{I}_b;\vb{a}](t) = \frac{l k_+  E_0 s}{k_+ s + k_- + l}\qty(1 - e^{-(k_+ s + k_- + l)t}),
\]
where $s$ is the concentration of \ce{S}. 
As $t \to \infty,$ the effective kinetics approaches to the classical form of the MM kinetics. This derivation is known as reactant stationary assumption in the literature \cite{schnell2014validity}.
\end{example}
\begin{example}[First-order network]
For a reaction network consisting of first-order reactions, equation (\ref{eqn_crn}) takes the form
\[
	\dv{\vb{x}}{t} = A\vb{x},
\]
where $A$ is a constant matrix.

Let the controlled value be $\vb{x}_b(t) \in \mathbb{R}^{\mathcal{S}_b}.$  We assume that the coordinates are set up so that the controlled values take indices $1,2, \cdots, k$. By splitting matrix $A$ into four blocks
\[
	A = \mqty( A_{bb} & A_{bi} \\ A_{ib} & A_{ii}),
\]
 we see that the equation of the internal species is
\[
	\dv{\vb{x}_i}{t} = A_{ii} \vb{x}_i + A_{ib} \vb{x}_b(t).
\]
The solution to this equation is
\[
	\vb{x_i}(t) = e^{t A_{ii}} \vb{a}_i + e^{t A_{ii}}\int_{0}^t e^{s A_{ii}} A_{ib} \vb{x}_b(s) \mathrm{d} s.
\]
The condition of consistency is
\[
	\dv{\vb{x}_b}{t} = A_{bi} \vb{x}_i(t) + A_{bb}\vb{x}_b(t),
\]
therefore the control has to be chosen depending on $\vb{a}_i$ unless $A_{bi} = 0$. In the physical sense, this consistency condition gives the flux resulting from the change of concentration.

Effective kinetics is given by
\[
	\mathcal{E}_s[\vb{I}_b;\vb{a}](t) = A_{si} \qty(e^{t A_{ii}} \vb{a}_i + e^{t A_{ii}}\int_{0}^t e^{s A_{ii}} A_{ib} \vb{x}_b(s) \mathrm{d} s )+ A_{sb}  \vb{x}_b(s),
\]
where $A_s = \qty(A_{sb}\, A_{si})$. 
\end{example}

For some systems, the average effective kinetics under constant concentration control becomes constant dependent on the concentration of the input species.  Phenomena of this kind can be observed, of course, if subnetworks with stable equilibria exist.  However, there is another pattern.
\begin{example}\label{ex_aek}
Let us consider the average effective kinetics for the network
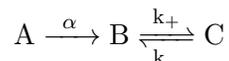
\begin{center}
\ce{A ->[\alpha] B <=>[k_+][k_{-}] C}
\end{center}
If we control the concentration of \ce{A}, the effective kinetics of \ce{C} is given by
\[
\mathcal{E}_C[\vb{I}_a;\vb{a}](t) = \alpha k_+ \int_0^t e^{-(k_+ + k_-)(t - s)} a(s)\mathrm{d} s,
\]
assuming $[C](0) = [B](0) =0.$ Thus the average effective kinetics of \ce{C} under constant concentration is 
\[
	\frac{\alpha k_+ [A]}{k_+ + k_-}.
\]
Note that there is no subnetwork with equilibria. In this case, we cannot apply the quasi-stationary state approximation. If we apply the quasi-stationary state approximation to \ce{B}, we obtain $\alpha [A]$ as the kinetics of \ce{C}.
\end{example}

\section{Straight-line solutions of non-autonomous ordinary differential equations}\label{rate}

In the rest of this paper, we consider the problem of the existence of a well-defined effective reaction rate under the assumption of constant input concentration. Such consideration of kinetics can be observed in the classical pseudo-first order approximation or the derivation of the Michaelis-Menten law by reactant stationary assumption. Here, we introduce the concept of straight-line solutions as a mathematical tool to formulate the notion of ``well-defined reaction rate" precisely. 
\subsection{Straight-line solutions of non-autonomous ordinary differential equations}\label{str}

The simplest solution of an ordinary differential equation is the equilibrium solution, which is defined by a zero of a vector field. The straight-line solution, which we define as follows, can be regarded as the next simplest kind of solution.
\begin{definition}[Straight-line solutions]
An ordinary differential equation $\dot{ \vb{x}} = \vb{F} \qty(t, \vb{x})$  admits a \textbf{straight-line solution} if there is a solution of form
\[
	\vb{x}(t) = t \vb{a} +\vb{c}
\]
where $\vb{a}$ and $\vb{c}$ are constant vectors. \end{definition}
In what follows, the coefficient $\vb{a}$ is called the \textbf{rate} of the straight-line solution.

We note that the existence of a straight-line solution is equivalent to that of a constant vector $\vb{c}$ with $\vb{F}\qty(\vb{F}\qty(\vb{c})t + \vb{c}) = \vb{F}\qty(\vb{c})$ for all $t$. In this sense, straight-line solutions are algebraically determined from the system.

\subsection{Affine systems}
First, we consider straight-line solutions for affine systems. Equations of this kind arise when we consider first-order networks with constant input, where the equation of kinetics typically takes the form
\[
	\dot{\vb{x}} = A \vb{x} + \vb{b}.
\]
For such systems, the notion of the straight-line solution characterizes the effective kinetics completely.
\begin{lemma}
Let $\vb{x}(t)$ be a solution of an affine ordinary differential equation $\dot{ \vb{x}} = A\vb{x} + \vb{b}$ such that there exists the limit
\[
	\lim_{t \to \infty} \dot{\vb{x}}(t) = \vb{a}.
\]
Then $\vb{a} \in \mathrm{ker}\, A$. Consequently, there exists a straight-line solution of the form $t \vb{a} +\vb{c}$, where $\vb{a} = A \vb{c} + \vb{b}$.
\end{lemma}
\begin{proof}
First, we observe that 
\[
	\lim_{t \to \infty} \frac{\vb{x}(t)}{t} = \vb{a}.
\]
This is established as follows. Let $\epsilon>0$ be arbitrary. Then, by the assumption, there exists $t_0 > 0$ such that
\[
	\|\dot{\vb{x}}(t) -\vb{a}\| < \epsilon
\]
for all $t > t_0$. Therefore we have
\[
	\begin{aligned}
		\left\| \frac{\vb{x}(t)}{t} -  \vb{a}\right\| &\leq \left\| \frac{1}{t}\int_{0}^{t} \dot{\vb{x}}(s) \mathrm{d} s -  \vb{a}\right\| + \frac{\left\|\vb{x}(0)\right\|}{t}\\
								&\leq \frac{1}{t}\qty(  \int_{0}^{t_0} \left\|\dot{\vb{x}}(s) -  \vb{a}\right\|\mathrm{d} s  +   \left\|\vb{x}(0)\right\|) + \frac{t-t_0}{t} \epsilon\\
								&\leq 2 \epsilon
	\end{aligned}
\]
for sufficiently large $t$.

From this result, we obtain
\[
	\begin{aligned}
		A \vb{a} & = \lim_{t \to \infty} \frac{A \vb{x}(t)}{t} \\
				&= \lim_{t \to \infty}  \frac{\dot{\vb{x}}(t) - \vb{b}}{t}\\
				&= 0.
	\end{aligned}
\]

Finally, $\vb{a} - \vb{b} \in \mathrm{im}\, A$ follows from the closedness of the image.
\end{proof}

 While we do not apply the results obtained here directly to chemical reaction networks, they illustrate the basic characteristics of the concept.
\begin{theorem}\label{thm_str}
An affine ordinary differential equation $\dot{ \vb{x}} = A\vb{x} + \vb{b}$ admits a straight-line solution for all choices of $\vb{b}$ if and only if the algebraic multiplicity of eigenvalue $0$ coincides with the dimension of the kernel of $A$. In addition, there is a non-constant straight-line solution if and only if $\vb{b} \not \in \mathrm{im}\, A.$ 
\end{theorem}
\begin{proof}
First we note that an affine ordinary differential equation $\dot{ \vb{x}} = A\vb{x} + \vb{b}$ admits a straight-line solution for all choice of $\vb{b}$ if and only if $\mathbb{R}^n = \mathrm{im}\, A + \mathrm{ker}\, A$. Indeed, $\vb{x}(t) = t \vb{a} +\vb{c}$ is a solution of the equation if and only if $\vb{a} = A\vb{c} + \vb{b}$ and $\vb{a} \in \mathrm{ker}\, A$. Consequently, the existence of straight-line solutions is equivalent to $\mathrm{ker}\, A \cap \mathrm{im}\, A = \{0\}$, which is exactly the condition for the algebraic multiplicity of eigenvalue $0$ to coincide with the dimension of the kernel of $A$. 

To show the second statement, first let $\vb{b} = A \vb{b}_1 + \vb{b}_2$, where $\vb{b}_2 \in \ker A$ is non-zero. Then $\vb{x}(t) = t \vb{b}_2 -\vb{b}_1$ is a non-constant straight-line solution. Conversely, if $\vb{b} = A \vb{b}_1$, then each straight-line solution $t \vb{a} +\vb{c}$ satisfies
\[
	\vb{a} = A\qty(\vb{c} + \vb{b}_1) \in \ker A,
\]
which implies $\vb{a} = \vb{0}.$ Therefore the straight-line solution is constant.
\end{proof}
\begin{remark}
If $\vb{b} \not \in \mathrm{im}\, A,$ there is no bounded solution. Indeed, the existence of eigenvalue $0$ implies that there is an invariant function of the form $f(\vb{x}) = \vb{w} \cdot \vb{x}$ for the ODE $\dot{\vb{x}} = A \vb{x}$, where $\vb{w}^T$ is a left eigenvector of $A$ with eigenvalue $0$. As is easily seen, we have $\dot{f} = \vb{w} \cdot \vb{b}$ for the case with input $\vb{b}$. Thus, the solutions are unbounded if $ \vb{w} \cdot \vb{b} \neq 0$.
\end{remark}
\begin{corollary}\label{cor_eig}
An affine ordinary differential equation $\dot{ \vb{x}} = A\vb{x} + \vb{b}$ admits a straight-line solution for all choices of $\vb{b}$ if the algebraic multiplicity of $0$ is one.  When $\vb{b} \not \in \mathrm{im}\, A$, the non-constant straight-line solution is unique up to the time-translation. In particular, the rate is uniquely determined. Further in addition, if all eigenvalues of $A$ has non-positive real parts, all solutions $\vb{x}(t)$ satisfies
\[
	\lim_{t \to \infty} \| \dot{\vb{x}} (t) - \vb{a}\| = 0,
\]
where $\vb{a}$ is the rate of the straight-line solution.
\end{corollary}
\begin{proof}
By Theorem \ref{thm_str}, there exists a straight-line solution under these assumptions. To show the uniqueness of the non-constant straight-line solutions, let $t\vb{a} + \vb{c}$ and $t\vb{a}' + \vb{c}'$ be solutions with non-zero rates. Then we have
\[
	A\qty(\vb{c}' - \vb{c}) = \vb{a}' - \vb{a} \in \ker A,
\]
which implies $\vb{a}' = \vb{a}$ because we have $\ker A \cap \mathrm{im}\, A = \{0\}$.  As $\ker A $ is spanned by $\vb{a}$, we have $\vb{c}' = \vb{c} + k \vb{a}$ for some $k$. Therefore, the non-constant straight-line solution exists and is unique up to the time-translation. 

For the last statement, we first note that $\mathrm{im}\, A$ is invariant under the flow generated by $A$. By assumption, $\vb{0}$ is a globally asymptotically stable equilibrium of the flow restricted to $\mathrm{im}\, A$. Let $ t \vb{a} +\vb{c}$ be a straight-line solution. Then, by direct calculations, we may show that every solution of the equation has the form
\[
	\vb{x}(t) = \vb{c} + t \vb{a} + e^{t A} \vb{y}_0
\]
for some $\vb{y}_0$. Then we have
\[
	\| \dot{\vb{x}} (t) - \vb{a}\| = \|e^{t A}A \vb{y}_0\|,
\]
which implies $\lim_{t \to \infty} \| \dot{\vb{x}} (t) - \vb{a}\| = 0$.
\end{proof}
The rate in Corollary \ref{cor_eig} can be calculated explicitly.
\begin{theorem}
Let $A$ be a matrix with eigenvalue $0$, which has algebraic multiplicity one, and other eigenvalues have negative real parts. Then the rate of the straight-line solution of $\dot{ \vb{x}} = A\vb{x} + \vb{b}$ is given by
\[
	\frac{1}{p'(0)} \mathrm{adj}\qty(-A) \vb{b},
\]
where $p(s)$ is the characteristic polynomial of $A$ and $\mathrm{adj}\qty(-A) $ is the adjugate matrix of $-A$.
\end{theorem}
\begin{proof}
By the assumption, we have $p'(0) \neq 0$ and there exists a polynomial $q$ such that $p(s) = s q(s),$ where $q(0) = p'(0).$
Let $\vb{x}(t)$ be a solution with initial value $\vb{x}_0$.
By Corollary \ref{cor_eig}, $\dot{\vb{x}}(t)$ is bounded. Therefore we have $\|\vb{x}(t)\| \leq M t + C$ for some positive constants $M$ and $C$. This enables us to consider the Laplace transform $\vb{X}(s)$ of $\vb{x}(t)$. Now we use the final value theorem to calculate the rate $\vb{a} = \lim_{t \to \infty} \dot{\vb{x}}(t)$.

By taking the Laplace transform, we obtain
\[
	\qty(sI - A)\vb{X}(s) = \vb{x}_0 +\frac{\vb{b}}{s}.
\]
Therefore we have
\[
	p(s) \vb{X}(s) =  \mathrm{adj}\,\qty(sI -A)\qty( \vb{x}_0 +\frac{\vb{b}}{s}).
\]
Noting that $A \,\mathrm{adj}\,\qty(sI -A) = -p(s)I + s\,  \mathrm{adj}\,\qty(sI -A),$ we obtain
\[
	\begin{aligned}
		s\qty(A\vb{X}(s) + \frac{\vb{b}}{s}) & = A\qty(\frac{p(s)}{q(s)} \vb{X}(s)) + \vb{b}\\
									&= \frac{1}{q(s)} \mathrm{adj}\,\qty(sI -A)\qty( \vb{x}_0 +\frac{\vb{b}}{s}) + \vb{b}\\
									&=  \frac{1}{q(s)}\qty(  -p(s)I + s\,  \mathrm{adj}\,\qty(sI -A))\qty( \vb{x}_0 +\frac{\vb{b}}{s}) + \vb{b}\\
									&=\frac{1}{q(s)} \mathrm{adj}\,\qty(sI -A))\qty( s\vb{x}_0 +\vb{b})-s \vb{x}_0. 
	\end{aligned}
\]
Therefore we have
\[
	\lim_{t \to \infty} \dot{\vb{x}}(t) = \frac{1}{p'(0)} \mathrm{adj}\qty(-A) \vb{b}.
\]
\end{proof}
A graph-theoretical sufficient condition is given as follows.
\begin{theorem}
Let the weighted digraph associated with matrix $A$ satisfy the following conditions:
\begin{enumerate}
	\item The loop of each vertex has a negative weight.
	\item Other than loops, each edge has a non-negative weight.
	\item The digraph is irreducible.
	\item The in-degree of each vertex is zero.
\end{enumerate}
Then, the affine ordinary differential equation $\dot{ \vb{x}} = A\vb{x} + \vb{b}$ admits a straight-line solution for all choices of $\vb{b}$. Further, all solutions $\vb{x}(t)$ satisfies
\[
	\lim_{t \to \infty} \| \dot{\vb{x}} (t) - \vb{a}\| = 0,
\]
where $\vb{a}$ is the rate of the straight-line solution.
\end{theorem}
\begin{proof}
Under these assumptions, the matrix $A$ is Metzler and irreducible. Therefore, there is a non-negative irreducible matrix $B$ and  $\alpha>0$ such that $A = -\alpha I + B.$ Further, it has $\vb{w} = (1,1, \cdots, 1)$ as a left eigenvector with eigenvalue $0$, which implies that the Perron-Frobenius eigenvector of $B$ is $\vb{w}$. As $\sigma(A) = \sigma(B) -\alpha$, we conclude that the algebraic multiplicity of eigenvalue $0$ is one for $A$, and all eigenvalues of $A$ have non-positive real parts. Therefore, Corollary \ref{cor_eig} is applicable.
\end{proof}

\subsection{Systems with the skew-product structures}
If a chemical reaction network has a group of species produced irreversibly, then the system of ordinary differential equations assumes a form called skew-product structure. 
\begin{definition}[Skew-product structure]
An ordinary differential equation $\dot{\vb{x}} = \vb{F}\qty(\vb{x})$ has the \textbf{skew-product structure} if there is a splitting of variables $\vb{x} = (\vb{p}, \vb{q})$ such that the equation assumes the following form:
\[
	\begin{aligned}
		\dv{\vb{p}}{t} &= \vb{f}\qty(\vb{p})\\
		\dv{\vb{q}}{t} &= \vb{g}\qty(\vb{p}, \vb{q})
	\end{aligned}
\]
\end{definition}
It is clear that the skew-product structure is dependent on the choice of coordinates. 

The next result gives conditions for an ODE with the skew-product structure to have a non-constant straight-line solution.
\begin{lemma}\label{lem_nec}
Let an ordinary equation $\dot{\vb{x}} = \vb{F}\qty(\vb{x})$ have the skew-product structure given by $\vb{F} =(\vb{f}, \vb{g})$. 
\begin{enumerate}
	\item If there is a straight-line solution, then the equation $\dot{\vb{p}} = \vb{f}\qty(\vb{p})$ has a straight-line solution.
	\item If there is a non-constant straight-line solution $t \vb{a} + \vb{c}$, then either $J_{\vb{p}} \vb{f}(\vb{c}_1)$ or $J_{\vb{q}} \vb{g}(\vb{c}_2)$ is degenerate, where $\vb{c} = \qty(\vb{c}_1, \vb{c}_2)$ and $J_{\vb{x}}\vb{G}(\vb{c})$ denotes the Jacobian matrix of $\vb{G}$ with respect to $\vb{x}$ evaluated at $\vb{c}$.
	\item If there is a globally asymptotically stable equilibrium, then the only straight-line solution is constant.
\end{enumerate}
\end{lemma}
An obvious scenario for (2) in the last lemma to happen is that $\vb{g}$ is independent of $\vb{q}$.
In this case, the existence of straight-line solutions can be verified easily. 
\begin{theorem}
Let an ordinary differential equation have the following form: 
\[
	\begin{aligned}
		\dv{\vb{p}}{t} &= \vb{f}\qty(\vb{p})\\
		\dv{\vb{q}}{t} &= \vb{g}\qty(\vb{p})
	\end{aligned}
\]
Assume that there is a globally asymptotically stable equilibrium for $\dv{\vb{p}}{t} = \vb{f}\qty(\vb{p})$. Then, there exists a straight-line solution and every solution $\vb{x}(t) = \qty(\vb{p}(t), \vb{q}(t))$ satisfies
\[
	\lim_{t \to \infty} \| \dot{\vb{q}} (t) - \vb{g}\qty(\vb{p}_{*}) \| = 0,
\]
where $\vb{p}_{*}$ is the equilibrium of $\dv{\vb{p}}{t} = \vb{f}\qty(\vb{p})$. The straight-line solution is unique up to the time-translation.
\end{theorem}
\begin{proof}
It is obvious that 
\[
	\vb{x}_*(t) = \mqty(\vb{p}_*\\ t \vb{g}\qty(\vb{p}_{*}))
\]
is a straight-line solution. The second assertion follows from
\[
	\|\dot{\vb{q}} (t) - \vb{g}\qty(\vb{p}_{*})\| = \| \vb{g}\qty(\vb{p}(t))- \vb{g}\qty(\vb{p}_{*})\|
\]
and the global asymptotic stability of $\vb{p}_{*}$. The uniqueness of the straight-line solution follows from Lemma \ref{lem_nec}.
\end{proof}
Moreover, we may calculate the time for the rate to stabilize if additional assumptions are satisfied.
\begin{corollary}
Let an ordinary differential equation have the following form: 
\[
	\begin{aligned}
		\dv{\vb{p}}{t} &= \vb{f}\qty(\vb{p})\\
		\dv{\vb{q}}{t} &= \vb{g}\qty(\vb{p})
	\end{aligned}
\]
Assume that there is a globally asymptotically stable equilibrium for $\dv{\vb{p}}{t} = \vb{f}\qty(\vb{p})$, $\vb{p}_*$ is uniformly exponentially stable i.e. there exist $\lambda, \gamma>0$ such that
\[
	\|\vb{p}(t) -\vb{p}_* \|\leq \gamma e^{-\lambda(t-t_0)}\|\vb{p}(t_0) -\vb{p}_*\|,
\]
and $\vb{g}$ is globally Lipschitz with Lipschitz constant $L$. Then, for all solution $\vb{x}(t) = \qty(\vb{p}(t), \vb{q}(t))$, we have
\[
	\| \dot{\vb{q}} (t) - \vb{g}\qty(\vb{p}_{*}) \|  \leq L \gamma e^{-\lambda(t-t_0)}\|\vb{p}(t_0) -\vb{p}_*\|.
\]
\end{corollary}
A more general case where the existence of the rate can be established is given by the following result, which can be obtained immediately from Theorem \ref{thm_str}.
\begin{corollary}
Let an ordinary differential equation have the following form: 
\[
	\begin{aligned}
		\dv{\vb{p}}{t} &= \vb{f}\qty(\vb{p})\\
		\dv{\vb{q}}{t} &= A \vb{q} + \vb{g}\qty( \vb{p}),
	\end{aligned}
\]
where $A$ is a real matrix.
Assume that there is an equilibrium for $\dv{\vb{p}}{t} = \vb{f}\qty(\vb{p})$. Then, a straight-line solution exists if the algebraic multiplicity of eigenvalue $0$ coincides with the dimension of the kernel of $A$.
\end{corollary}

\section{Skew product systems and quasi-stationary state approximations}\label{skew}
In this section, we consider the effective reaction rates for reaction networks with the skew-product structure. Such a situation is rather common. If a kinetic reaction network $(\mathcal{S}, \mathcal{R}, \vb{s}, \vb{t}, \mathcal{K})$ has a species $P \in \mathcal{S}$ as an irreversible product, that is, it appears only as a product, then we may write the equation of kinetics
in a skew product form
\[
	\begin{aligned}
		\dv{\vb{x}}{t} &= \vb{f}\qty(\vb{x})\\
		\dv{[P]}{t} &= g\qty(\vb{x}).
	\end{aligned}
\]
When we regard a chemical reaction network as a conversion system, some species are contained only as reactants, and others appear only as products. 
\begin{definition}
Let $(\mathcal{S}, \mathcal{R}, \vb{s}, \vb{t})$ be a chemical reaction network. 
\begin{enumerate}
	\item A species $F \in \mathcal{S}$ is a \textbf{feeding species} if $F \in \vb{s}(\rho)\cup \vb{t}(\rho)$ implies $F \in \vb{s}(\rho)$.
	\item A species $P \in \mathcal{S}$ is a \textbf{product species} if $P \in \vb{s}(\rho)\cup \vb{t}(\rho)$ implies $P \in \vb{t}(\rho)$.
	\item A species $I \in \mathcal{S}$ is an \textbf{intermediate species} if there are reactions $\rho$ and $\rho'$ with $I \in \vb{s}(\rho)$ and $I \in \vb{t}(\rho')$ .
\end{enumerate}
\end{definition}
In the classical procedure of quasi-stationary state approximation, a subnetwork of a chemical reaction network is assumed to be in a stationary state. This amounts to considering the equilibrium solution for some species, and the overall reaction rate is determined in terms of them. Here, we would like to consider such an approximation scheme for the case of networks with product species. With this in mind, we define the quasi-stationary state approximation as follows.
\begin{definition}
An ordinary differential equation with the skew product structure
\[
	\begin{aligned}
		\dv{\vb{p}}{t} &= \vb{f}\qty(\vb{p})\\
		\dv{\vb{q}}{t} &= \vb{g}\qty(\vb{q}, \vb{p})
	\end{aligned}
\]
admits \textit{quasi-stationary state approximation (QSSA)} if there is $\vb{p}_*$ such that 
\[
	 \vb{f}\qty(\vb{p}_*) = 0.
\]
\end{definition}
\begin{remark}
It is usual for QSSA to assume slow-fast dichotomy of time scales \cite{segel1989quasi, stiefenhofer1998quasi, wechselberger2020geometric}. This distinction is not apparent in the treatment here, nor do we assume the steady state's stability. However, we may regard $\vb{p}$ as the fast variable in mind. 
\end{remark}
An important example is an ordinary differential equation with the skew product structure and an additive input term $\vb{k}$
\[
	\begin{aligned}
		\dv{\vb{p}}{t} &= \vb{f}\qty(\vb{p}) + \vb{k}\\
		\dv{\vb{q}}{t} &= \vb{g}\qty(\vb{q}, \vb{p})
	\end{aligned}
\]
admits QSSA if there is $\vb{p}_*$ such that 
\[
	 \vb{f}\qty(\vb{p}_*)+ \vb{k} = 0.
\]

If QSSA is applicable, the average effective kinetics of \ce{P} is given by $g(\vb{x}_0)$, where $\vb{f}(\vb{x}_0) = 0$.

Reaction networks with additive inputs may arise in the following way.
\begin{definition}
For a kinetic reaction network $(\mathcal{S}, \mathcal{R}, \vb{s}, \vb{t}, \mathcal{K})$, $\mathcal{F} \subset \mathcal{S}$ is \textbf{additive input} if
\begin{enumerate}
	\item  Each $F \in \mathcal{F}$ is a feeding species.
	\item For all $\rho \in \mathcal{R}$, $\mathcal{F}\cap \vb{s}(\rho) \neq \emptyset$ implies $\vb{s}(\rho) \subset \mathcal{F}$.
\end{enumerate}
In particular, $\mathcal{F}$ is \textbf{first-order input}  if $\mathcal{F}\cap \vb{s}(\rho) \neq \emptyset$ implies $\vb{s}(\rho) = F$ for some $F \in \mathcal{F}$.

A product species $P \in \mathcal{S}$ is a \textbf{first-order output} if $P \in \vb{t}(\rho)$ implies $\vb{t}(\rho) = P$.
\end{definition}

Then, it is easy to observe that the networks with additive inputs induce a non-autonomous equation with an additive input term if we control the concentration of input species. Conversely, equations of reaction networks with additive inputs can be regarded as that of a reaction network where first-order inputs are added.

The QSSA is applicable for a rather large class of networks with first-order inputs and outputs.
\begin{theorem}\label{thm_qssa}
Let $(\mathcal{S}, \mathcal{R}, \vb{s}, \vb{t}, \mathcal{K})$ be a mass-action reaction network and a species $P \in \mathcal{S}$ be first-order output. If the network satisfies the following conditions, then QSSA is applicable for the kinetics of $P$ when the concentration of feeding species is controlled:
\begin{enumerate}
	\item If a species $F$ never appears as a product, then $F$ is a first-order input.
	\item If a species $X$ never appears as a reactant, then $X$ is a first-order output.
	\item If reactions of the form \ce{0 ->[\alpha_i] F_i} are added and the outputs are substituted with zero complexes, the resulting network has deficiency zero and is weakly reversible.
\end{enumerate}
\end{theorem}
\begin{proof}
By assumptions, the equation of kinetics assumes the following form:
\[
	\begin{aligned}
		\dv{\vb{f}}{t} & =  - K\vb{f}\\
		\dv{\vb{p}}{t} & = B \vb{x}\\
		\dv{\vb{x}}{t} & = \tilde{S} \vb{r}\qty(\vb{x})+K \vb{f},\\
	\end{aligned}
\]
where $K$ is diagonal and invertible.
The equation of kinetics of the network in (3) is given by
\[
	\begin{aligned}
		\dv{\vb{f}}{t} & =  -K \vb{f}+ \vb{\alpha}\\
		\dv{\vb{x}}{t} & = \tilde{S} \vb{r}\qty(\vb{x})+K \vb{f}.\\
	\end{aligned}
\]
Since this modified network is weakly reversible and deficiency zero, it admits a positive equilibrium $\qty(\vb{f}_0(\vb{\alpha}),\vb{x}_0(\vb{\alpha}))$ for all $\vb{\alpha}$. As $\vb{\alpha}$ is arbitrary and $K$ is invertible, $\vb{f}_0$ is arbitrary. Therefore, for all $\vb{f}_0$, there exists $\vb{x}_0$ such that
\[
	0 = \tilde{S} \vb{r}\qty(\vb{x}_0)+K \vb{f}_0,
\]
which implies that the skew product system
\[
	\begin{aligned}
		\dv{\vb{p}}{t} & = B \vb{x}\\
		\dv{\vb{x}}{t} & = \tilde{S} \vb{r}\qty(\vb{x})+K \vb{f},\\
	\end{aligned}
\]
admits the QSSA.
\end{proof}

\begin{corollary}
Let $(\mathcal{S}, \mathcal{R}, \vb{s}, \vb{t}, \mathcal{K})$ be a mass-action reaction network, a species $P \in \mathcal{S}$ be first-order output, and a species $F \in \mathcal{S}$ be first-order input. If the network satisfies the following conditions, then QSSA is applicable  for the kinetics of $P$ under the controlled concentration of feeding species:
\begin{enumerate}
	\item If a species $F$ never appears as a product, then $F$ is a first-order input.
	\item If a species $X$ never appears as a reactant, then $X$ is a first-order output.
	\item Terminal strong-linkage classes consist of first-order outputs. For the network obtained by reversing all reactions, terminal strong-linkage classes consist of first-order inputs.
	\item If reactions of the form \ce{0 -> F_i} are added and the outputs are substituted with zero complexes, the resulting network has deficiency zero.
\end{enumerate}
\begin{proof}
By assumption (3), every complex $C$ ultimately reacts to a first-order output. By considering the reversed network, we see that every complex $C$ is contained in a path starting from a first-order input to a first-order output. Therefore, the modified network in (4) is weakly reversible.
\end{proof}
\end{corollary}
We recall that a reaction network $(\mathcal{S}, \mathcal{R}, \vb{s}, \vb{t}, \mathcal{K})$ is \textbf{conservative} if the cokernel of the stoichiometric matrix contains a positive vector. For conservative networks, we may calculate the result of QSSA using flux balance analysis.
\begin{lemma}
Let $(\mathcal{S}, \mathcal{R}, \vb{s}, \vb{t}, \mathcal{K})$ be a kinetic reaction network and $P \in \mathcal{S}$ is an irreversible product of the network. If the network is conservative with conserved charge $\vb{\mu} = \qty(\vb{\mu}_0, \mu_P),$ the result of QSSA under an additive constant input $\vb{k}$ is given by 
\[
	\dv{[P]}{t} = \frac{\vb{\mu}_0}{\mu_P} \vb{k},
\]
provided QSSA is applicable.
\end{lemma}
\begin{proof}
By assumptions, the kinetics of the network is governed by an equation of the form
\[
		\begin{aligned}
		\dv{\vb{x}}{t} &=\tilde{S}\vb{r}\qty(\vb{x}) + \vb{k}\\
		\dv{[P]}{t} &= g\qty(\vb{x}),
	\end{aligned}
\]
where $\vb{x}$ is the vector of the concentration of species other than $P$, $\tilde{S}$ is the corresponding stoichiometric matrix. We note that the total stoichiometric matrix $S$ of the network  is
\[
	S = \mqty(\tilde{S} \\ \vb{u}^T),
\]
where $\vb{u}$ is a non-negative vector. If QSSA is applicable, there is concentration profile $\vb{x}_0$ such that
\[
	\tilde{S}\vb{r}\qty(\vb{x}_0) + \vb{k} = 0.
\]
By the definition of conserved charge, we have
\[
	\vb{\mu}^T S \vb{r} \qty(\vb{x}_0)  = \tilde{\vb{\mu}}^T\tilde{S}\vb{r}\qty(\vb{x}_0)  + \mu_P \vb{u}^T \vb{r} \qty(\vb{x}_0) = 0.
\]
Noting that $g\qty(\vb{x}) = \vb{u}^T \vb{r} \qty(\vb{x}),$ we have
\[
	\dv{[P]}{t} = -  \frac{1}{\mu_P}\tilde{\vb{\mu}}^T\tilde{S}\vb{r}\qty(\vb{x}_0) = \frac{\tilde{\vb{\mu}}^T}{\mu_P} \vb{k}.
\]
\end{proof}

By specializing this result to the case where additive inputs occur naturally, we obtain the following result.
\begin{theorem}
Let $(\mathcal{S}, \mathcal{R}, \vb{s}, \vb{t}, \mathcal{K})$ be a kinetic reaction network and $P \in \mathcal{S}$ is an irreversible product of the network. If $\mathcal{F} = \{F_1, F_2, \cdots, F_{k}\}$ is  a set of first-order inputs and the network obtained by deleting $\mathcal{F}$ and reactions starting from $\mathcal{F}$ is conservative, the result of QSSA is given by first-order reactions:
\[
	\dv{[P]}{t} = R_1 [F_1] + R_2 [F_2] + \cdots + R_k [F_k],
\]
provided QSSA is applicable. In particular, if the total mass of the network except $\{F_1, F_2, \cdots, F_{k}\}$ is conserved, the production rate of $P$ coincides with the intake rate of $\mathcal{F}$ by the network.
\end{theorem}
\begin{proof}
In this case, we may apply the result to the network obtained by deleting $\mathcal{F}$ and reactions starting from $\mathcal{F}$ and input $A \vb{f}$, where $A$ is a matrix denoting the intake rate and $\vb{f}$ is the concentration vector of the species in $\mathcal{F}$. 
\end{proof}
\begin{corollary}
For a reaction network satisfying the assumptions of Theorem \ref{thm_qssa}, the result of QSSA is given by first-order reactions.
\end{corollary}

Thus, as far as the intake rate is first order, the effective reaction at the steady state appears to be first order regardless of detailed reaction mechanisms.

In the setting of the theorem, nonlinear coupling between species affects the applicability of QSSA rather than the linearity of the effective rate.

\begin{example}
Let us consider a network
\ce{F ->[k_1] X}, \ce{F ->[k_2] Y}, \ce{X + Y ->[l] P}. The kinetics has a skew product structure:
\[
\begin{aligned}
	\dv{[X]}{t}  &= -l [X] [Y] + k_1 [F]\\
	\dv{[Y]}{t}  &= -l [X] [Y] + k_2 [F]\\
	\dv{[P]}{t}  &= l [X] [Y].
\end{aligned}
\]
It is easy to observe that QSSA is not applicable if $k_1 \neq k_2$. If $k_1 = k_2$, the effective production rate of \ce{P} is $k_1 [F]$.
\end{example}

\section{Summary and conclusions}\label{conc}
In this article, we have considered the effective kinetics of a chemical reaction network defined as a conversion rate of feeding species into product species and studied its basic properties. As is observable from the comparison with QSSA, approximation or calculation of the rate of a reaction mechanism can be carried out within the framework of effective kinetics in a unified fashion. The discussion here mainly pertained to the problem of obtaining a well-defined reaction rate under the assumption that the concentration of the feeding species is kept constant. While we barely mentioned it here, it appears to be interesting to consider also the response of chemical reaction networks under varying inputs using the framework of non-autonomous dynamical systems. In particular, the existence of the rate will be closely related to that of an attractor, and in this way, we may obtain a clearer picture of how the output of a reaction system stabilizes.
\section*{Acknowledgements}
This work was supported by Kakenhi Grant-in-Aid for Scientific Research (22H05153, 23K19021) and RIKEN Incentive Research Grant. The author would like to express deep gratitude to Dr. Hideshi Ooka and Prof. Ryuhei Nakamura for their support and fruitful discussion. 

 \bibliography{effective_kinetics.org.tug} 
\bibliographystyle{junsrt} 

\end{document}